\documentclass[letterpaper, 10 pt, conference]{ieeeconf}  % Comment this line out if you need a4paperhttps://www.overleaf.com/project/639b9b1445738307b0f49d84

\IEEEoverridecommandlockouts                              % This command is only needed if % you want to use the \thanks command

\usepackage[utf8]{inputenc}
\usepackage{amsmath}
\usepackage{graphicx}
\usepackage{amsfonts}
\usepackage{xcolor}
\usepackage{amsmath,amssymb,amsfonts}

\newtheorem{proposition}{Proposition}

\newtheorem{remark}{Remark}
\newtheorem{definition}{Definition}

\newtheorem{asmpt}{Assumption}
\usepackage{subcaption} % for subfigures
\usepackage[ruled,vlined]{algorithm2e}
\usepackage{glossaries}
\newacronym{rl}{RL}{reinforcement learning}
\newacronym{ml}{ML}{machine learning}
\newacronym{ai}{AI}{artificial intelligence}
\newacronym{sgd}{SGD}{stochastic gradient descent}
\newacronym{marl}{MARL}{multi-agent reinforcement learning}
\newacronym{dl}{DL}{deep learning}
\newacronym{nn}{NN}{neural network}
\newacronym{gnn}{GNN}{graph neural network}
\newacronym{gat}{GAT}{Graph Attention Network}
\newacronym{ctde}{CTDE}{centralized training decentralized execution}
\newacronym{dtde}{DTDE}{decentralized training decentralized execution}
\newacronym{gcn}{GCN}{graph convolutional network}
\newacronym{posg}{POSG}{partially observable stochastic game}
\newacronym{dpomdp}{Dec-POMDP}{decentralized partially observable Markov decision process}
\newacronym{pomdp}{POMDP}{partially observable Markov decision process}
\newacronym{manet}{MANET}{mobile ad-hoc network}
\newacronym{vanet}{VANET}{vehicular ad-hoc network}
\newacronym{mlp}{MLP}{multi-layer perceptron}
\newacronym{uav}{UAV}{unmanned aerial vehicle}
\newacronym{lstm}{LSTM}{Long Short-Term Memory}
\newacronym{mas}{MAS}{multi-agent systems}
\newacronym{peb}{PEB}{partial episode bootstrapping}
\newacronym{lti}{LTI}{linear time-invariant}
\newacronym{lvc}{LVC}{live, virtual, and constructive}
\usepackage{csquotes} % Proper quotes

\IEEEoverridecommandlockouts 

\usepackage{todonotes}

\title{\bf Safety-Aware Multi-Agent Learning for Dynamic Network Bridging}

\author{Raffaele Galliera, Konstantinos Mitsopoulos, Niranjan Suri, and Raffaele Romagnoli% <-this % stops a space
\thanks{R. Romagnoli is with the School of Science and Engineering, Department of Mathematics and Computer Science, Duquesne University, 600 Forbes Ave, Pittsburgh, PA 15213, USA:
        {\tt\small romagnolir@duq.edu}}%
\thanks{R. Galliera, N. Suri, and K. Mitsopoulos are with the Institute for Human and Machine Cognition, 40 South Alcaniz St, Pensacola, FL 32502, USA:
       {\tt\small  kmitsopoulos@ihmc.org }%
       {\tt\small  rgalliera@ihmc.org }
       {\tt\small  nsuri@ihmc.org }%
       }%
}

\begin{document}

\maketitle
\begin{abstract}

Addressing complex cooperative tasks in safety-critical environments poses significant challenges for multi-agent systems, especially under conditions of partial observability. We focus on a dynamic network bridging task, where agents must learn to maintain a communication path between two moving targets. To ensure safety during training and deployment, we integrate a control-theoretic safety filter that enforces collision avoidance through local setpoint updates. We develop and evaluate multi-agent reinforcement learning safety-informed message passing, showing that encoding safety filter activations as edge-level features improves coordination. The results suggest that local safety enforcement and decentralized learning can be effectively combined in distributed multi-agent tasks.
\end{abstract}

\section{Introduction}
% \gls{rl} has made significant progress in addressing complex problems, from robot control to game play, in recent years. Many real-world applications, such as autonomous vehicles, \gls{uav} operations, and robot interactions, inherently involve multiple agents, thereby necessitating the development of \gls{marl} approaches. These approaches address the challenges posed by increasing joint action spaces and decentralization requirements due to the rise in the number of agents. Although \gls{marl} demonstrates promising results in simulations, its deployment in real-world, safety-critical applications introduces substantial safety challenges.
Despite the promising performance of \gls{marl} in simulations, its application in real-world, safety-critical scenarios raises significant concerns. A key challenge is collision avoidance, especially in applications like autonomous driving, where failure to prevent collisions could result in severe consequences. Moreover, the inherent nature of \gls{mas} often involves partial observability, further increasing the complexity of these problems by limiting the information available to each agent about the environment and the states of other agents. 

This challenge highlights a fundamental limitation of all \gls{rl}-based approaches: while agents learn to optimize predefined reward functions, their behavior lacks guarantees and remains inherently unpredictable. The sole reliance on \gls{marl}'s reward functions to ensure safety is insufficient \cite{elsayed2021safe}. Even with extensive training and monitoring of an agent's performance based on expected rewards, or other performance measures, it is impossible to exhaustively examine all potential scenarios in which the agent might fail to act safely or predictably.

% Safe RL approaches aim to learn policies that maximize expected rewards while satisfying safety constraints. Despite the proposition of several methods, including learning from demonstrations and reward-shaping, ensuring safety during initial learning phases remains a challenge. Shielding frameworks have been proposed to address this issue by synthesizing shields that monitor and adjust agents' actions to ensure safety during learning. However, designing effective shields for real-world applications, especially in scenarios with centralized shielding, faces challenges due to scalability limitations and communication delays.

In ensuring safety among agents, various control-theoretic approaches and hybrid learning approaches have been developed, ranging from Model Predictive Control (MPC) \cite{dai2017distributed} to Reference Governor (RG) techniques \cite{li2021safe}, and from Control Barrier Functions (CBFs) \cite{gao2023online} to Reachability Analysis (RA) methods \cite{kochdumper2023provably}. 
In robotics applications, those methods are known as safety filters \cite{hsu2023safety}, and they are deployed to ensure safety regardless of the algorithm used for the agent to accomplish a task. While safety filters can effectively prevent unsafe behavior, they are often conservative by design. This conservatism may lead to conflicts with task objectives, especially in \gls{mas} where coordination is essential.

In this work, we develop a training framework for \gls{marl} that incorporates safety constraints during policy learning. We adapt an optimization-free safety tracking control algorithm developed in \cite{romagnoli2023software}, which is based on ellipsoidal positively invariant sets. We show how this safety filter can be adapted to the context of \gls{mas} and how safety is guaranteed. We integrate this type of filter in the agents' training process, allowing them to learn decentralized policies and avoid situations where the filter compromises task objectives.

% Unfortunately, safety and task requirements may not always align, and we illustrate this aspect with our proposed safety filter, which provides a conservative approach that may preclude the success of the task. Instead of defining a protocol that resolves those situations, we provide a safety-aware training for our \gls{marl} system that is able to avoid situations where safety conditions strongly impact task requirements.

% The safety analysis for each single agent is carried out by considering our agents as linear time-invariant (LTI) systems using Lyapunov theory.

%In this work, we focus on developing a hybrid learning approach built upon the method used in \cite{romagnoli2023software}, that can learn to achieve complex objectives while ensuring safety guarantees in distributed \gls{mas}. 
% Although this method was initially developed to ensure safety against cyber-attacks, in this paper, we show that it is instrumental in ensuring safety in the context of partially observable multi-agent systems.

% To demonstrate and evaluate the effectiveness of our hybrid approach we consider the cooperative task of Dynamic Network Bridging~\cite{galliera_2024} (Fig. \ref{fig:task}). In this task the goal is to establish and maintain a connection between two moving targets A and B in a 2D plane, relying on a swarm of $N$ agents taking decentralized decisions with limited observation of the environment. The agents must form an ad-hoc mobile network that establishes a communication path between the two targets as they move, dynamically adjusting their positions.

\begin{figure}[ht]
    \centering
    \includegraphics[scale=0.3]{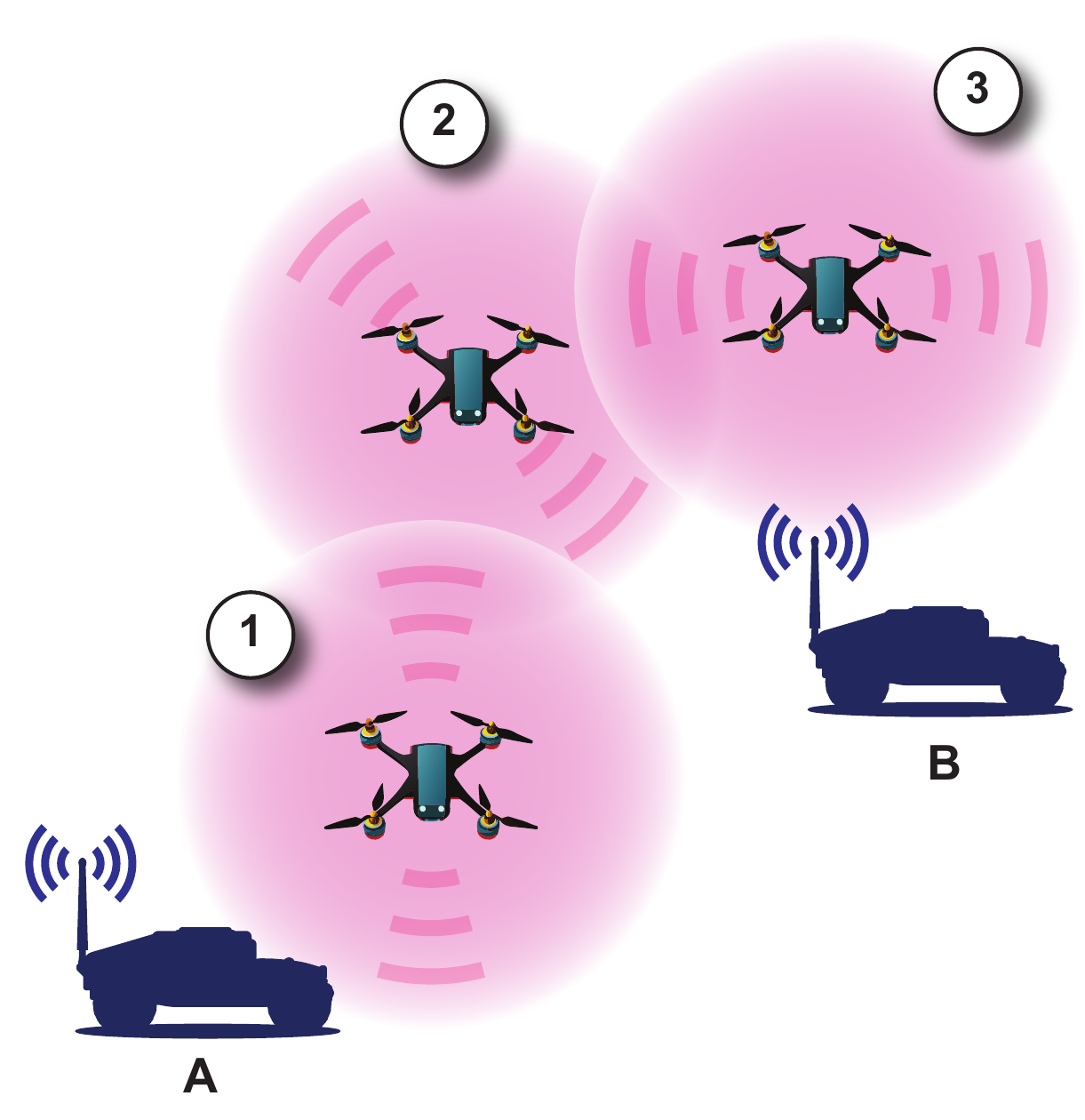}
    \caption{Decentralized swarm coordination for dynamic network bridging}
    \label{fig:task}
\end{figure}

% Each agent operates with a localized view, perceiving only its one-hop neighboring agents and the relative positions of the targets within its limited sensor range. The agents must collaborate in a decentralized manner to develop effective multi-agent control policies in order to a) satisfy the dynamic network bridging objective while b) respecting their physical boundaries and avoid collisions.

% to address the challenge of coordinating a swarm of agents under partial observability, to

% In this paper we introduce a hybrid approach combining \gls{marl} and control-theoretic methods, accomplishing task objectives while providing safety guarantees.
We demonstrate our approach in the dynamic network bridging task~\cite{galliera_2024} (Fig. \ref{fig:task}) where a number of agents establish and maintain a communication path between two moving targets. Our key contributions are as follows: 1) A novel decentralized control framework that integrates a safety filter based on invariant sets into the \gls{marl} agents' training process, and restricts the effect of each agent's movement updates to only its one-hop neighbors, enabling efficient local coordination; 2) We augment the learning process with safety-informed message passing, where agents observe safety filter activations as edge-level features during training, leading to improved coordination 3) An algorithm that updates setpoints while preserving safety conditions through communication with affected neighbors; 3) An analytical computationally tractable condition verifying potential safety violations during updates. Overall our approach provides effective coordination among agents while ensuring safety constraints.

% First, we propose a decentralized control framework that integrates a safety filter based on ellipsoidal invariant sets into the MARL training process. This allows agents to coordinate using only local one-hop communication while avoiding unsafe configurations. Second, we introduce a setpoint update algorithm that ensures safety through local checks and provides an analytical condition for detecting potential safety violations. Third, we apply this framework to a dynamic network bridging task involving mobile agents and targets, showing that agents can learn effective decentralized policies that satisfy task objectives without triggering safety interventions.

% Our contributions are threefold. First, we propose a decentralized control framework compatible with \gls{marl}, where each agent’s setpoint updates are coordinated locally considering their limited communication range. Second, we introduce an efficient setpoint update algorithm that guarantees safety through local interactions and provides a tractable condition to check for potential safety violations. Third, we apply this approach to a dynamic network bridging task involving mobile agents and targets, demonstrating that our framework enables safe coordination without compromising task performance.

%  We develop and evaluate multi-agent reinforcement learning safety-informed message passing, showing that encoding filter activations as edge-level features improves coordination.
 
\section{Preliminaries}

\subsection{Task description}
Dynamic network bridging is a multi-agent task focused on addressing the problem of dynamically establishing and maintaining a communication link between two mobile, independent nodes (targets) using a swarm of autonomous agents taking decentralized decisions~\cite{galliera_2024}. Agents and targets are moving in a 2D grid space and have a limited communication range, within which they exchange information. This region can be defined as a 2D ball of radius \( r \), centered at its position \( p_i(t) = [p_{x,i}(t), p_{y,i}(t), p_{z,i}(t)]^T \). Based on this definition, two agents $i$ and $j$ can communicate if and only if their regions intersect:
\begin{equation}\label{1hop}
\mathcal{B}_r(p_i(t)) \cap \mathcal{B}_r(p_j(t)) \neq \emptyset.
\end{equation}
where \[
\mathcal{B}_r(p_i(t)) = \left\lbrace p \in \mathbb{R}^3 : \|p - p_i(t) \| ^2 \leq r^2 \right\rbrace
\]
In our setup, we consider $p_{z,i}$ fixed, therefore the one-hop communication region can be described by a 2D ball.
Formally, the system is represented as a dynamic graph $\mathcal{G}(t) = (\mathcal{V}, \mathcal{E}(t))$ with nodes $\mathcal{V}$ corresponding to agents and targets. An edge $(u, v)$ exists between nodes if their Euclidean distance at time-step $t$ is within a predefined communication range.

The task reflects scenarios where communication infrastructure is absent or unreliable, such as in disaster zones or remote environments. Agents must position themselves to form an ad-hoc, decentralized communication network that connects the two mobile targets. Each agent operates based on local observations and limited-range communication, with no access to global state information.

\subsection{Multi-agent reinforcement learning}
In a typical multi-agent cooperative scenario, $N$ agents interact within an environment to complete a task. Noisy and limited sensors may prevent each agent from observing the full state of the environment, allowing access only to partial observations. With partial observability, an agent no longer knows the true state of the environment, but instead needs to maintain a belief state - a probability distribution over states estimated from its observation history - which it uses to select actions. Decision-making processes in such contexts can be formulated as a \gls{dpomdp} \cite{bernstein2002complexity, Oliehoek_Amato_2016}.

A \gls{dpomdp} is defined by a tuple $(S, \mathbf{A}, \mathcal{O}, T, R, \Omega, \gamma, N)$ where $S$ is a finite set of states of the environment, $\mathbf{A}$ is a set of joint action spaces, and $N$ is the number of agents, $T$ is the state transition probability function which gives the probability of transitioning from state $s$ to state $s^{\prime}$ when joint action $\mathbf{a}=(a_1, \cdots, a_N)$ is taken; $R$ is the joint reward function that maps states and joint actions to real numbers and is used to specify the goal of the agents; $\Omega$ is a set of joint observations $O_i$, where $O_i$ is the observation set for agent $i$; $\mathcal{O}$ is the joint observation probability function, which gives the probability of receiving joint observation $\mathbf{o}=(o_1, \cdots, o_N)$ after taking joint action $\mathbf{a}$ and ending up in state $s^{\prime}$; and $\gamma$ is the reward discount factor, $0 \leq \gamma \leq 1$.
The goal of the agents is to maximize a joint reward function. Solving a \gls{dpomdp} optimally is computationally intractable, therefore, various approximate solution methods \cite{oliehoek2013approximate, lowe2017multi} have been used, such as heuristic search, dynamic programming, value or policy-based approximators with neural networks etc.

% In a \gls{dpomdp}, agents need to collaborate and coordinate their actions based on their individual observations to maximize the joint reward function. Solving a \gls{dpomdp} optimally is computationally intractable, as the policy space grows exponentially with the number of agents and the planning horizon. Therefore, various approximate solution methods \cite{oliehoek2013approximate, lowe2017multi} have been used, such as heuristic search, dynamic programming, value or policy-based approximators with neural networks etc.
% The problem facing the team \konsshort[]{wording here} is to find the optimal joint policy, i.e. a combination of individual agent policies that produces behavior that maximizes the team’s expected reward.

\subsection{Distributed autonomous swarm formation for dynamic network bridging}
\label{subsec:dnb}

We formulate our problem as a \gls{dpomdp}, where each agent has restricted sensing capabilities given by a specific communication range, leading to incomplete information about the global state of the system outside such area. However, whenever two entities (agents or targets) are within range, a communication channel becomes available for them to exchange information. This channel enables the possibility for agents to know each other features and execute distributed strategies requiring message exchange to coordinate. At the beginning of every time-step, agents make decentralized decisions, generating target points $w$, utilizing their local observation to choose the direction of their next movement. To solve the \gls{dpomdp} that corresponds to our task we used a \gls{gnn} approximator that encodes the graph structure of the problem.

\subsection{Agent Dynamics}
We consider a \gls{lti} of the form  
\begin{equation}\label{eqn:lti_sys}
\dot{x} = A (x - x_{sp})
\end{equation}  
where \( x \in \mathbb{R}^n \) is the state of the system, and the origin is translated to the setpoint \( x_{sp} \in \mathbb{R}^n \). The matrix \( A \in \mathbb{R}^{n \times n} \) is Hurwitz (i.e., \( \mathrm{Re}(\lambda_i) < 0 \)), where \( \lambda_i \in \mathbb{C} \) is the \( i \)-th eigenvalue of \( A \). The matrix \( A \) can be seen as the closed-loop system matrix obtained by applying a state-feedback controller (i.e., \( u = -K(x - x_{sp}) \)), where \( u \) is the control input.  

For the asymptotically stable \gls{lti} system \eqref{eqn:lti_sys}, there exists a quadratic Lyapunov function of the form  
\begin{equation}\label{lyap_fun}
V(x) = (x - x_{sp})^T P (x - x_{sp})
\end{equation}  
where \( P \) is a symmetric positive definite matrix that satisfies the Lyapunov equation  
\begin{equation}\label{eqn:lyapunov}
A^T P + P A = -Q
\end{equation}  
for a given symmetric positive definite matrix \( Q > 0 \).  

Defining the \( P \)-norm as  
\begin{equation}\label{P-norm}
\lVert x - x_{sp} \rVert_P \triangleq \sqrt{(x - x_{sp})^T P (x - x_{sp})},
\end{equation}  
for a fixed value \( c > 0 \), the Lyapunov function \eqref{lyap_fun} defines an ellipsoid, which can be written in terms of \eqref{P-norm} as follows:  
\begin{equation} \label{pinvset}
\mathcal{E}_{c}(x_{sp}) = \{ x \in \mathbb{R}^n \mid \lVert x - x_{sp} \rVert_P^2 \leq c \}.
\end{equation}  

From \eqref{eqn:lyapunov}, \( V(x(t)) \) is a monotonically decreasing function; therefore, for any \( x(t) \) starting within the ellipsoid \eqref{pinvset}, the trajectory never leaves it. Hence,  
\( \mathcal{E}_c \) is called a \textit{positively invariant set} \cite{Khalil:1173048}. For the rest of the paper, we will refer to the state of the agent $i$ as $x_i$ and for the related setpoint as $x_{sp,i}$. In this context, the tracking algorithm updates the setpoint of the agent to reach a desired position. This formulation can be extended to more general nonlinear systems by using linearization \cite{Khalil:1173048}, or Lyapunov-like methods such as quadratic boundedness \cite{brockman1995quadratic}, as used in the context of UAVs in \cite{romagnoli2023software}, where ellipsoidal positively invariant sets are provided. We use this last case in our experiments.

\subsection{Safety}
%\rafG[]{Should we expand a little more on what we mean by safety?}
%Safety considerations are crucial for the practical deployment of autonomous swarms. Thus, we introduce explicit safety conditions which can be triggered to prevent inter-agent collisions:

\begin{definition}\label{def:safety}
Given two agents $i$ and $j$, with distinct equilibrium points $x_{sp,i} \neq x_{sp,j}$ and initial states $x_i(t_0) \in \mathcal{E}_c(x_{sp,i})$ and $x_j(t_0) \in \mathcal{E}_c(x_{sp,j})$, respectively, agents $i$ and $j$ are deemed \textit{safe} relative to each other if:
\begin{equation}\label{safe1}
\mathcal{E}_c(x_{sp,i}) \cap \mathcal{E}_c(x_{sp,j}) = \emptyset.
\end{equation}
\end{definition}

This definition leverages the invariant properties of \eqref{pinvset}, ensuring agents' trajectories remain confined within their respective safety ellipsoids, thereby preventing collisions.

To operationalize this safety criterion within the two-dimensional communication space, we impose the following:
\begin{asmpt}\label{asm:projection}
The projection of $\mathcal{E}_c(x_{sp,i})$ onto the 2-D plane defined by the coordinates $(p_x, p_y)$ must always lie within the 2-D communication ball $\mathcal{B}_r(p_i(t))$, where $p_i(t) = [p_{x,i}(t), p_{y,i}(t)]$ represents the actual position of agent $i$.
\end{asmpt}

\section{Safety Filter for Tracking Control}

The proposed algorithm consists of two steps: first, computing a feasible safe step (Figure~\ref{fig:tracking}); and second, checking whether the new position would interfere with another agent (Figure~\ref{fig:safety}).
% the setpoint update of the agent $i$ that guarantees that its current position belongs to $\mathcal{E}_c$ of the new setpoint. This condition is fundamental for the second step where the safety condition \eqref{safe1} is checked for all agents within the one-hop communication ball. 

Let us start with the \textbf{first step} which has been borrowed from \cite{romagnoli2023software}. For each agent, we select two scalars \(c\) and \(s\) such that $0 < s < c $ therefore $\mathcal{E}_{s}(x_{sp,i})\subset \mathcal{E}_{c}(x_{sp,i})$.

\begin{definition}\label{def:two}
We can say that the agent has reached the setpoint \(x_{sp,i}\) if and only if $x_i(t) \in \mathcal{E}_s(x_{sp,i})$.  
\end{definition}

Assuming that the agent $i$ is moving from the target point $w_i$ to the target point $w_{i+1}$, generated by its decentralized policy, the idea is to generate a sequence of setpoints on the line that connects the two target points. Once the state of agent $i$, $x_i$ reached the setpoint \(x_{sp,i}\), a new setpoint \(x'_{sp,i}\) is generated following the rule
\begin{equation}\label{update}
x'_{sp,i} = x_{sp,i} + (\sqrt{c} - \sqrt{s})v
\end{equation}
where \(v \in \mathbb{R}^n\) with \(\| v \|_2 = 1\), indicating the direction of the line connecting the two target points. As showed in \cite{romagnoli2023software}

\begin{equation}\label{cond:safe1}
x_i(t) \in \mathcal{E}_s(x_{sp,i}) \Rightarrow x_i(t) \in \mathcal{E}_c(x'_{sp,i}).    
\end{equation}

Fig. \ref{fig:tracking} illustrates the process of updating the setpoint, presented in 2-D space for clearer visualization, though it actually corresponds to a state space with
n=12 components, as applicable to quadrotors.

\begin{figure}[h]
    \centering
    \begin{subfigure}{0.25\textwidth} % Adjust width as needed
        \centering
        \includegraphics[trim=15 0 10 45, clip,width=\textwidth]{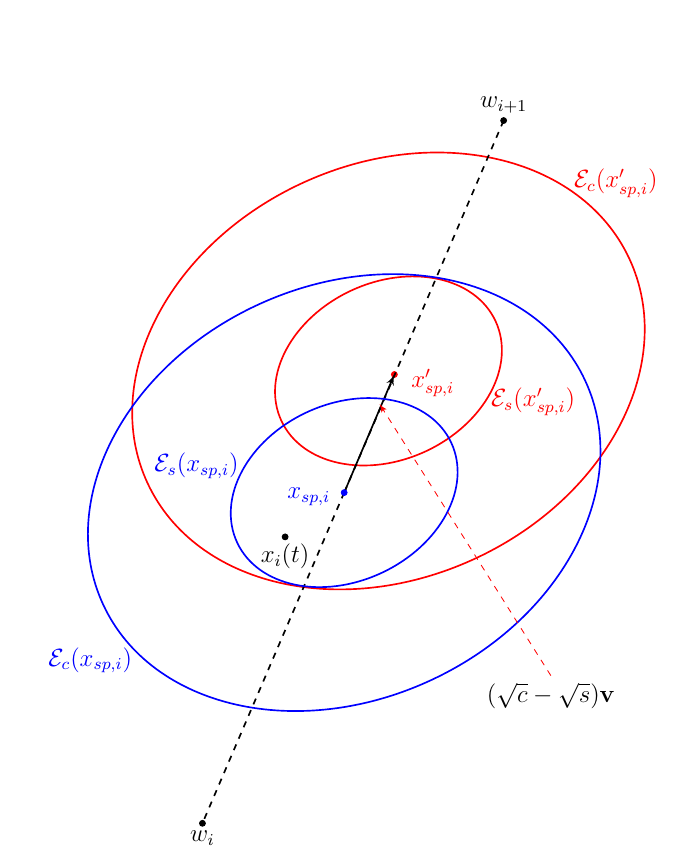}
        \caption{Setpoint update for agent $i$ in the state space.}
        \label{fig:tracking}
    \end{subfigure}
    \hfill
    \begin{subfigure}{0.2\textwidth} % Adjust width as needed
        \centering
        \includegraphics[width=\textwidth]{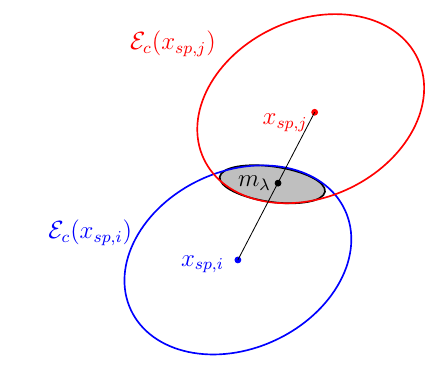}
        \caption{ $m_\lambda$ is on the segment that connects the two setpoints.}
        \label{fig:safety}
    \end{subfigure}
    \caption{Safe tracking control and safety verification.}
    \label{fig:combined}
\end{figure}

% \begin{figure}[!ht]
%     \centering
%     %\includegraphics[width=0.8\columnwidth]{figure/safety_cond.pdf}
%     \includegraphics[trim=15 0 10 45, clip,scale=0.4]{figure/safety_cond.pdf}
%     \caption{Setpoint update for agent $i$ in the state space. Note that the target points $w_i$ and $w_{i+1}$ have the same structure of the setpoints $x_{sp,i}$.}
%     \label{fig:tracking}
% \end{figure}
Thanks to \eqref{update} and Assumption \ref{asm:projection}, the projection of $\mathcal{E}_c(x'_{sp,i})$ is contained in $\mathcal{B}(p_i(t))$. 
\begin{proposition}\label{prop:one}
    Considering agent $i$ moving from target point $w_i$ to $w_{i+1}$, and updating the setpoint $x_{sp,i}$ every time $x_{i}(t)\in \mathcal{E}_s(x_{sp,i})$ by using \eqref{update}. Considering also Assumption 1 true, then the safety condition \eqref{safe1} can be violated only with all the agents $j\neq i$ that satisfy \eqref{1hop} (i.e. within the one-hop communication range). 
\end{proposition}
\begin{proof}
    From Assumption \ref{asm:projection}, the ball $\mathcal{B}_r(p_i(t))$ encompasses the projection onto the 2-D space of $\mathcal{E}_c(x_{sp,i})$ for any current state $x_i(t) \in \mathcal{E}_c(x_{sp,i})$. Utilizing \eqref{update} to adjust the setpoint ensures that $x_i(t)$ remains within $\mathcal{E}_c(x'_{sp,i})$, signifying that its 2-D space projection still falls within $\mathcal{B}_r(p_i(t))$. Thus, should the updated setpoint $x'_{sp,i}$ contravene \eqref{safe1}, such a violation would only occur concerning agents situated within the one-hop communication range.
\end{proof}
%----here

Proposition \ref{prop:one} highlights a critical aspect of setpoint updates: setpoint updates only impact safety for agents within communication range. This ensures agent $i$ can exchange messages with affected neighbors before executing updates. This allows coordinating to preemptively address any potential safety concerns from the new setpoint through direct communication with impacted agents.

Building upon Proposition \ref{prop:one}, let's define $\mathcal{D}_i(t)$ as the set comprising all agents $j \neq i$ that are within the one-hop communication range of agent $i$ at time $t$. The safe setpoint update algorithm for each agent $i$ is given in Algorithm \ref{algo:setpoint_update}. Once agent $i$'s state $x_i(t)$ reaches its current setpoint $x_{sp,i}$ (as per Definition \ref{def:two}), it computes a new target setpoint $x'_{sp,i}$. Before updating, agent $i$ checks for any possible safety violations by evaluating the safety condition \eqref{safe1} for all agents in $\mathcal{D}_i(t)$. If no violation is detected, $x_{sp,i}$ is updated to the new $x'_{sp,i}$; otherwise, the previous setpoint is retained.

\begin{algorithm}
\SetAlgoLined
\KwIn{Agent $i$ with setpoint $x_{sp,i}$, $\mathcal{D}_i(t)$}
\KwOut{Updated setpoint $x_{sp,i}$}
\BlankLine
 potential\_safety\_violation $\leftarrow$ 0\;
\If{$x_i(t)\in \mathcal{E}_s(x_{sp,i})$}{ 
    Compute $x'_{sp} \leftarrow$ \eqref{update}\;
    \For{$j\in\mathcal{D}_i(t)$}{
        \If{$\mathcal{E}_c(x_{sp,i})\cap \mathcal{E}_c(x_{sp,j})\neq \emptyset$ }{
        potential\_safety\_violation $\leftarrow 1$\;
        break\;
     }\    
}\
\If{potential\_safety\_violation $==$ $0$}{
     $x_{sp,i}\leftarrow x'_{sp,i}$\;
     }
}
\caption{Setpoint Update Algorithm}
\label{algo:setpoint_update}
\end{algorithm}

% In the case the distance between the current and the updated setpoint is greater thant the distance between the current set point and the target point $w_{i+1}$ then the new setpoint coincide with the target point:

% \begin{equation}\label{target_pt}
%  \Vert x_{sp,i} - x'_{sp,i}\Vert > \Vert x_{sp,i} - w_{i+1}\Vert, \Rightarrow x'_{sp}=w_{i+1}   
% \end{equation}

% \begin{asmpt}\label{asm:two}
%      The initial setpoints of all the agents satisfy the safety condition \eqref{safe1}. 
% \end{asmpt}
% \begin{asmpt}\label{asm:three}
%      During the setpoint update for agent $i$, all the agent $j\in D_i(t)$ cannot change their setpoints $x_{sp,j}$. 
% \end{asmpt}

% \begin{proposition}\label{prop:two}
% Considering the multi-agent system described with \eqref{MAS} and the Assumptions 1, 2, and 3, then Algorithm 1 guarantees the safety condition \eqref{safe1} is satisfied for all the agents.
% \end{proposition} 
% \begin{proof}
%     Given that all agents have setpoints that satisfy the safety condition \eqref{safe1}, Algorithm \ref{algo:setpoint_update} updates each agent's setpoint only if the safety condition is adhered to within the set of communicating agents, as stated in Proposition \ref{prop:one}. To ensure that the current setpoint remains a safe solution at all times, we rely on Assumption \ref{asm:two}. Furthermore, Assumption \ref{asm:three} prevents agents within one-hop communication range from modifying the setpoint while any agent executes Algorithm \ref{algo:setpoint_update}.
% \end{proof}

By applying Algorithm \ref{algo:setpoint_update} to each agent, we ensure that the safety condition \eqref{safe1} is satisfied throughout the system. This is guaranteed by the fact that an agent can only update its setpoint after verifying that the new target does not violate safety for any agents within its one-hop communication range, as stated in Proposition \ref{prop:one}. The proof of safety relies on the following assumptions:

\begin{asmpt}\label{asm:two}
    The initial setpoints of all agents satisfy the safety condition \eqref{safe1}.
\end{asmpt}

\begin{asmpt}\label{asm:three}
    During the setpoint update for agent $i$, all agents $j \in \mathcal{D}_i(t)$ cannot change their setpoints $x_{s p, j}$.
\end{asmpt} 

\begin{proposition}
    Considering the \gls{mas} described where the agents are \gls{lti} systems  \eqref{eqn:lti_sys} and Assumptions \ref{asm:projection}, \ref{asm:two} and \ref{asm:three}, then applying Algorithm \ref{algo:setpoint_update} to each agent guarantees that the safety condition \eqref{safe1} is satisfied for all agents.
\end{proposition}

\begin{proof}
    By Assumption \ref{asm:two}, the initial setpoints of all agents satisfy the safety condition \eqref{safe1}. Whenever an agent $i$ updates its setpoint using Algorithm \ref{algo:setpoint_update}, it checks if the new target setpoint $x_{s p, i}^{\prime}$ violates the safety condition for any agent $j \in \mathcal{D}_i(t)$. If a violation is detected, the setpoint is not updated, ensuring that the current safe solution is maintained. Proposition \ref{prop:one} guarantees that safety is only impacted for agents within the one-hop communication range, which agent $i$ can directly communicate with before executing the update. Furthermore, Assumption \ref{asm:three} prevents agents in $\mathcal{D}_i(t)$ from changing their setpoints during agent $i$ 's update, ensuring a consistent safety evaluation. Therefore, by applying Algorithm \ref{algo:setpoint_update} to each agent, the safety condition \eqref{safe1} is upheld throughout the system's operation.
\end{proof}

%here
%$\Vert x_{sp,i} - x'_{sp,i}\Vert > \Vert x_{sp,i} - w_{i+1}\Vert$, then $x'_{sp}=w_{i+1}$.

%Considering two agents \(i\) and \(j\), and at time \(t\), agent \(i\) updates the setpoint with \(x'_{sp,i}\), the safety condition is
%\[
%\mathcal{E}_c(x'_{sp,i}) \cap \mathcal{E}_c(x_{sp,j}) = \emptyset 
%\]
%\begin{itemize}
%\item Since agents \(i\) and \(j\) can communicate, agent \(i\) knows \(x_{sp,j}\) and can verify whether or not the safety condition is violated.
%\item Different strategies can be implemented in the case of safety violation. The simplest one is to keep the old \(x_{sp,i}\).
%\end{itemize}
In the \textbf{second step}, we need to verify whether the safety condition \eqref{safe1} is going to be violated. Our objective is to demonstrate that if \eqref{safe1} is violated, then there exists a point on the line connecting $x_{sp,i}$ and $x_{sp,j}$, inside the overlap region of two ellipsoids $\mathcal{E}_c(x_{sp,i})$ and $\mathcal{E}_c(x_{sp,j})$ (Fig. \ref{fig:safety}). This overlap region, denoted by $\mathcal{E}_c(x{sp,i}) \cap \mathcal{E}_c(x_{sp,j})$, satisfies the following condition:
\begin{equation}\label{intersect_cond}
\lambda(x-x_{sp,i})^TP(x-x_{sp,i}) + (\lambda - 1)(x-x_{sp,j})^TP(x-x_{sp,j}) \leq c
\end{equation}
for \(0 \leq \lambda \leq 1\) with \(x_{sp,i} \neq x_{sp,j}\).

\begin{proposition}\label{prop:three}
Let us consider the ellipsoids \(\mathcal{E}_c(x_{sp,i})\) and \(\mathcal{E}_c(x_{sp,j})\) defined as in \eqref{pinvset}. The set of points \(x\) satisfying \eqref{intersect_cond} for all \(\lambda \in [0,1]\) is either empty, a single point, or an ellipsoid:
\begin{equation}\label{intersect_ellips}
\hat{\mathcal{E}}_{K_\lambda}(m_\lambda) = \left\lbrace x \in \mathbb{R}^n : (x-m_\lambda)^TP^{-1}(x-m_\lambda) \leq K_\lambda \right\rbrace
\end{equation}
where
\begin{align*}
m_\lambda &= \lambda x_{sp,i} + (1-\lambda)x_{sp,j}\\
K_\lambda &= 1 - \lambda(1-\lambda)(x_{sp,j} - x_{sp,i})^TP(x_{sp,j} - x_{sp,i}).
\end{align*}
\end{proposition}
\begin{proof}
    See \cite{gilitschenski2012robust}.
\end{proof}
\vspace{0.2cm}
Proposition \ref{prop:three} shows that if there is an ellipsoid violation, then the point $m_\lambda$ belongs to the intersection Fig. \ref{fig:safety}:
\[
\mathcal{E}_c(x_{sp,i})\cap \mathcal{E}_c(x_{sp,j}) \neq \emptyset \rightarrow m_\lambda \in \mathcal{E}_c(x_{sp,i})\cap \mathcal{E}_c(x_{sp,j})
\]

% \begin{figure}[ht!]
%     \centering
%     %\includegraphics[width=0.8\columnwidth]{figure/safety_cond_a.pdf}
%     \includegraphics[scale=0.7]{figure/safety_cond_b.pdf}
%     \caption{Explaination of Proposition 3: $m_\lambda$ is on the segment that connects the two setpoints.}
%     \label{fig:safety}
% \end{figure}

Therefore, to test if the two ellipsoids intersect we need to find if there exists at least one point on the segment joining $x_{sp,i}$ and $x_{sp,j}$ that satisfies \eqref{intersect_cond}. To do so we consider both ellipsoids
\begin{equation}\label{two_ellips}
\left\lbrace\ \begin{array}{l}
(x-x_{sp,i})^TP(x-x_{sp,i})\leq \rho\\
(x-x_{sp,j})^TP(x-x_{sp,j})\leq \rho\\
\end{array}  \right.
\end{equation}

We replace $x$ with $m_\lambda$ since we are checking only a point in the segment joining the two setpoints.  Defining $d \triangleq x_{sp,j}- x_{sp,i}$,  \eqref{two_ellips} can be rewritten as
\begin{equation}
\left\lbrace\ \begin{array}{r}
\lambda^2 d^TPd\leq \rho\\
(\lambda -1)^2d^TPd\leq \rho\\
\end{array}  \right. \Rightarrow \left\lbrace\ \begin{array}{r}
\lambda^2 \Vert d \Vert_P^2\leq \rho\\
(\lambda -1)^2\Vert d \Vert_P^2\leq \rho\\
\end{array}  \right.
\end{equation}
Solving for $\lambda$ we obtain
\begin{equation}
\left\lbrace\ \begin{array}{l}
\lambda^2 \leq \frac{c}{\Vert d \Vert_P^2}\\
\lambda^2 \Vert d \Vert_P^2-2\lambda \Vert d \Vert_P^2+\Vert d \Vert_P^2\leq c\\
\end{array}  \right. \Rightarrow 1-2\sqrt{\frac{c}{\Vert d \Vert_P^2}}\leq 0
\end{equation}
The two ellipsoids $\mathcal{E}_c(x_{sp,i})$ and $\mathcal{E}_c(x_{sp,j})$ intersect if
\begin{equation}\label{int_cond}
2\sqrt{\frac{c}{\Vert d \Vert_P^2}}\geq 1  \rightarrow \Vert d \Vert_P^2\leq 4c   
\end{equation}

\begin{remark}
    Although Algorithm 1 ensures safety according to Proposition 2, it does not guarantee target achievement. Agents within one-hop communication range may become immobilized at a certain position because any attempt to reach the target points could potentially violate the safety condition \eqref{safe1}. We intentionally refrain from employing any enforcer to resolve deadlock scenarios. Instead, we allow the agents to learn how to avoid or recover from them through the training process.

\end{remark}

In summary, our analysis provides the following key theoretical results:
\begin{itemize}
    \item Each agent's setpoint update only impacts the safety of neighboring agents within its one-hop communication range. This localized effect is an important property that enables decentralized coordination.
    \item we establish that Algorithm \ref{algo:setpoint_update} guarantees the preservation of the safety condition in (\ref{safe1}) for all agents,
    \item We provide \eqref{int_cond} a computationally tractable method for evaluating possible safety condition \eqref{safe1} violations.
\end{itemize}
Together, these theoretical results form the basis for our safe and distributed multi-agent control framework, ensuring that agents can dynamically update their setpoints while maintaining the prescribed safety guarantees through local communication and coordination.
    
%    \hfill
%    \begin{subfigure}
%        \centering
%        \includegraphics[width=0.5\columnwidth]{figure/safety_cond_a.pdf}
%        \caption{Subfigure B}
%        \label{fig:subfigb}
%    \end{subfigure}
%    \caption{Example of subfigures}
%    \label{fig:subfigures}
%\end{figure}

\begin{table*}[h]
    \centering
    \begin{tabular}{l|ccc|cc}
    \hline
    & \multicolumn{3}{c|}{\textbf{Training Conditions}} & \multicolumn{2}{c}{\textbf{Evaluation Results}} \\
    \hline
    & Safety Filter & Explicit Penalty & Truncation & Avg Coverage & Avg Safety Interventions \\
    \hline
    Baseline & No  & No & No  & $(37.63 \pm 16.62)\%$ & $26.99 \pm 12.31$ \\
    \textbf{Approach A} & \textbf{Yes} & \textbf{No}  & \textbf{No}  & $(50.14 \pm 12.72)\%$ & $20.45 \pm 7.49$ \\
    Approach B & Yes & Yes & No  &  $(45.02 \pm 11.81)\%$  & $15.25 \pm 7.82$ \\
    Approach C & Yes & Yes & Yes &  $(10.86 \pm 6.44)\%$ & $9.76 \pm 5.72$ \\
    \hline
\end{tabular}

    \caption{Comparison of the trained agents performance across 100 evaluation episodes.}
    \label{tab:agent_performance_safety_truncation}
\end{table*}

\section{Decentralized Safety-Aware Agents}

This section details the components of our approach to ensure safe interactions while learning effective positioning strategies in networked and decentralized multi-agent tasks, by combining \gls{marl} with our safety filter mechanism. The core intuition is to employ Algorithm \ref{algo:setpoint_update} since the training phase and inform the agents of eventual safety filter activations.

\subsection{Guiding learning via decentralized safety filters}
By enforcing adherence to explicit safety conditions defined in \eqref{safe1}, agents experience realistic operational constraints and encounter possible deadlock scenarios whenever direct trajectories toward their target points trigger the safety filter. These constraints inherently shape the policy optimization landscape during training: 1) creating deadlock conditions agents might be blocked by their respective safety filters unless they reconsider their actions; 2) avoiding unsafe states of the environment, which are not going to be visited by the agents. Consequently, the agents need to leverage interactions and communication with neighboring agents to collaboratively avoid or resolve deadlock situations by generating alternative, safety-compliant paths to achieve their objectives effectively.

\subsection{Safety-informed edges}
\label{subsec:safety_informed_edges}
% \kons[]{seems complementatry with the preliminaries where we describe the features. also i think we do not need to state again what kind of environments we will use.
% Also, we should clearly state what is the final observation set that each agent receives.}
% We consider networked multi-agent environments, such as dynamic network bridging (Section~\ref{subsec:dnb}), where each agent's observation is limited by its communication range, within which it can perceive and communicate with other entities. This local perspective naturally enables defining edge-level features to capture relational information critical for effective decentralized decision-making.

Leveraging the graph nature of the task introduced in Section~\ref{subsec:dnb}, we explicitly introduce edge-level features linked to the conditions determining the activation of the decentralized safety filter. Specifically, edge features comprise base attributes controlling the activation of the safety filter and a binary safety activation indicator between connected agents. Formally, given two connected agents $i$ and $j$, the edge features are defined as:
\begin{equation}
    f_l^{ij} = [\mathbf{f}, S_{i,j}],
\end{equation}
where $\mathbf{f}$ denotes a vector of available attributes related to the activation of the decentralized safety filter, and the binary indicator $S_{i,j}$ explicitly signals whether the safety filter has been activated between the two agents, thereby preventing setpoint update.

\subsection{Policy parameterization}
To demonstrate our approach we train the agents with a simple independent Q-Learning algorithm~\cite{watkins1992q} and employ the same \gls{gnn} architecture proposed in \cite{galliera_2024} for action-value estimation. The architecture integrates \gls{gat}~\cite{veličković2018graph} and \gls{lstm} layers~\cite{HausknechtS15}, enabling effective modeling of spatial and temporal dependencies essential for collaborative decision-making. The \gls{gat} layers allow agents to dynamically attend to and integrate neighbor information, significantly improving inter-agent coordination by adaptively sharing latent representations through message passing. The \gls{lstm} captures temporal dependencies and partial observability, allowing agents to leverage historical context in their decisions. The architecture is trained using a \gls{ctde} framework, where the agents optimize the same action-value function parameterization during training while acting independently based on local observations and the interaction with the other agents.

\section{Experimental Setup}
\label{sec:exp_setup}
In this section, we describe in detail the specific environment implementation and training configurations employed for our \gls{marl} agents.

\subsection{Dynamic network bridging environment implementation}

\textbf{Environment:} The simulated environment is normalized to a $[0,1] \times [0,1]$ grid, though agents can move beyond these limits. Each entity has a communication radius of 0.20, which defines the area within which it can detect and communicate with others.

\textbf{Agent observations:}
At each time step, agents receive an observation consisting of their local one-hop graph structure. This includes \textbf{node-level} features for themselves and all directly connected neighbors, as well as \textbf{edge-level} features for each neighbor pair. An example is shown in Figure~\ref{fig:features}, where Agent A observes Agent B. The observation includes both their node features $f_n^A, f_n^B$ and the corresponding edge features $f_i^{AB}$ from a successful connection. Node-level features include the agent's ID, current coordinates, last action taken, and the positions of the targets A and B. These features provide basic local context for decision-making. Edge-level features reflect spatial and safety-related information, as introduced in Section~\ref{subsec:safety_informed_edges}. For each neighboring agent $j$, agent $i$ observes relative position offsets $\Delta p_x = p_{x,i} - p_{x,j}$, $\Delta p_y = p_{y,i} - p_{y,j}$, and their Euclidean distance $d_{i,j}=\sqrt{\Delta p_x^2+\Delta p_y^2}$. Including safety-informed edge features allows agents to adjust their actions based not only on position but also on the state of the safety filter. This supports more informed and coordinated decision-making during training and execution.

% An example of agent observation is illustrated in Figure~\ref{fig:features}. Agent A perceives Agent B, leading its observation to include their respective node-level feature ($f_n^A, f_n^B$) and their edge-level features ($f_i^{AB}$).  

% At the beginning of each time-step, the agents receive an observation of their local one-hop graph structure, along with the node-level features of each node connected to them. Such features include the node ID, current coordinates, action taken, and coordinates of the targets A and B to be connected. Following our methodology introduced in Section~\ref{subsec:safety_informed_edges}, we explicitly incorporate edge-level features directly tied to the activation conditions of our decentralized safety controller. Given the spatial nature of the environment and the operational logic of our safety filter, the base feature vector $\mathbf{f}$ includes relative positional information given by the Cartesian offsets $\Delta p_x = p_{x,i} - p_{x,j}$, $\Delta p_y = p_{y,i} - p_{y,j}$, and their Euclidean distance $d_{i,j}=\sqrt{\Delta p_x^2+\Delta p_y^2}$. The binary indicator $S_{i,j}$ explicitly denotes whether the safety filter (Algorithm~\ref{algo:setpoint_update}) is activated for the agent pair, thereby preventing setpoint update.

% The integration of safety-informed edge features enriches the agents' relational context, enabling them to condition their actions based on safety-related information.

\textbf{Agent actions:} At the beginning of every time-step, agents make decentralized decisions, generating target points $w$, utilizing their local observation to choose the direction of their next movement. This is encoded in a two-dimensional action space, with each dimension having three options: move forward, move backward, or hold along the corresponding $p_x$ or $p_y$ axis. Given the direction chosen by the agents, a new target point $x_{sp}$ is calculated for each agent using a fixed offset equal for both the $p_x$ and $p_y$ axes.

\textbf{Reward:} Agents are rewarded using a reward function that motivates agents to form a connected network bridging the targets. Such reward signal combines three components: 

\textit{Base connectivity reward:}
The base reward, $R_{\text{base}}$, encourages larger connected components:
\begin{equation}
    R_{\text{base}}(s) = \frac{|C_{\text{max}}(s)|}{|\mathcal{V}|},
\end{equation}
where $|C_{\text{max}}(s)|$ represents the size of the largest connected component in state $s$, and $|\mathcal{V}|$ is the total number of entities.

\textit{Centroid distance penalty:
}The centroid distance penalty, $P_{\text{cent}}$, is based on agents' positions relative to the targets, penalizing them based on the Euclidean distance between the centroid of the agents' positions and the central point between the targets.

\textit{Target path bonus:}
A bonus, $B_{\text{path}}=100$, is awarded if a path exists between the two targets.

The overall reward combines these three elements:

\begin{equation}
    \label{eq:final_reward}
    R(s, a) = \begin{cases} 
    
    B_{\text{path}}(s) & \mbox{if } \exists \mathrm{path}(T_1, T_2);
    
    \\ R_{\text{base}}(s) - P_{\text{cent}}(s), & \mbox{otherwise}.
    \end{cases} 
\end{equation}

\textbf{Simulated UAV dynamics:}
For our experiments, we use a six-degrees-of-freedom (6-DOF) \gls{uav} as an agent, where its dynamics, control laws, and safety sets used in our simulations are defined in the Appendix of \cite{romagnoli2023software}. For the \gls{uav} simulations, the nonlinear model of the drone is used.

\textbf{Target mobility:} To increase task complexity, each target moves deterministically toward waypoints located outside the communication range of any agent. When a target reaches a waypoint, it pauses for a fixed number of time-steps (5 in our implementation) before selecting a new one. Waypoint selection ensures that at least two agents (and up to $N$) are required to maintain a communication bridge between the targets.

\begin{figure}[ht]
    \centering
    \includegraphics[scale=0.3]{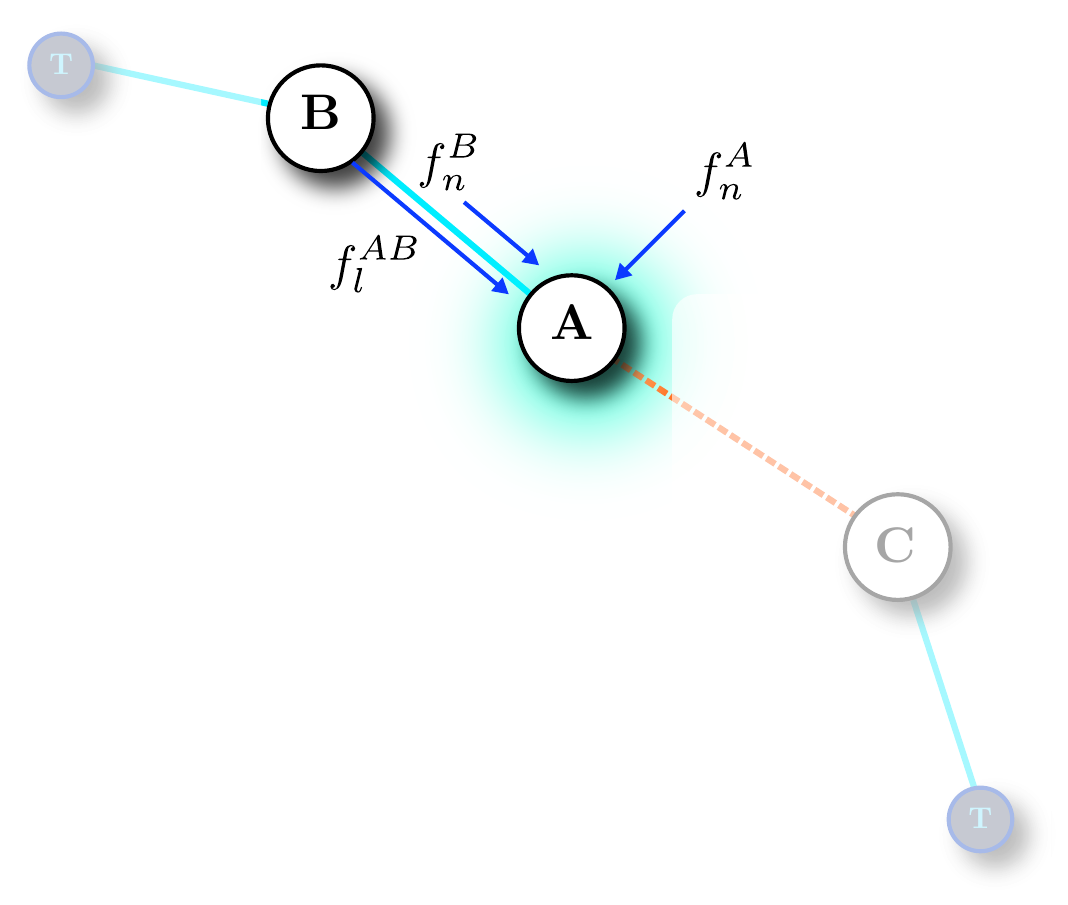}
    \caption{A typical configuration of our system. Agents are represented with letters A, B, C and targets with T. Agent A and B have established a successful connection. Agent A receives node features $f_n^A$, $f_n^A$ and edge features $f_n^{AB}$ as inputs.}
    \label{fig:features}
\end{figure}

\subsection{Training configurations}
We train our agents using episodes of 100 time steps, for a total of one million steps. For generalization, we initialize agent positions, target locations, and target motion patterns using different seeded random states across training and testing scenarios. During evaluation, each trained agent is tested in 100 new episodes, and we measure performance metrics with the safety filter consistently enabled.

We systematically evaluated different training configurations to assess the impact of our safety mechanism:

\begin{itemize}
\item \textbf{Baseline:} A vanilla \gls{marl} approach with no safety filter enabled. Edge features include only the base features $\mathbf{f}$.
\item \textbf{Approach A (Safety Filter):} Our main contribution. Agents use Algorithm \ref{algo:setpoint_update} to prevent safety infringements. The binary safety activation indicator is incorporated into the agents' edge-level features ($\mathbf{f}$ and $S_{i, j}$).
\item \textbf{Approach B (Safety + Penalty):} We add an explicit penalty of $-10$ to Approach A whenever an agent attempts an action that would trigger the safety filter.
\item \textbf{Approach C (Safety + Penalty + Early Truncation):} Building on Approach B, we truncate episodes immediately upon safety-triggering situations, thus enforcing even stricter safety adherence.
\end{itemize}

Every approach adopts the same learning algorithm and \gls{gnn} architecture. In any configuration employing the safety filter, an agent violating \eqref{safe1} has its trajectory toward $x_{sp}$ blocked at its current position.

Furthermore, we train an alternative of Approach A, called \textbf{Approach A-Node}, to investigate the benefits of our proposed integration of safety-informed edges~\ref{subsec:safety_informed_edges}. In this alternative approach, the agents utilize only the node-level features with the addition of a binary indicator $\mathbb{S}$ in case the agent had his safety filter activated during the previous step.

\section{Results}

\begin{figure}[t]
    \centering
    \begin{subfigure}[t]{0.48\textwidth}
        \centering
        \includegraphics[scale=0.3]{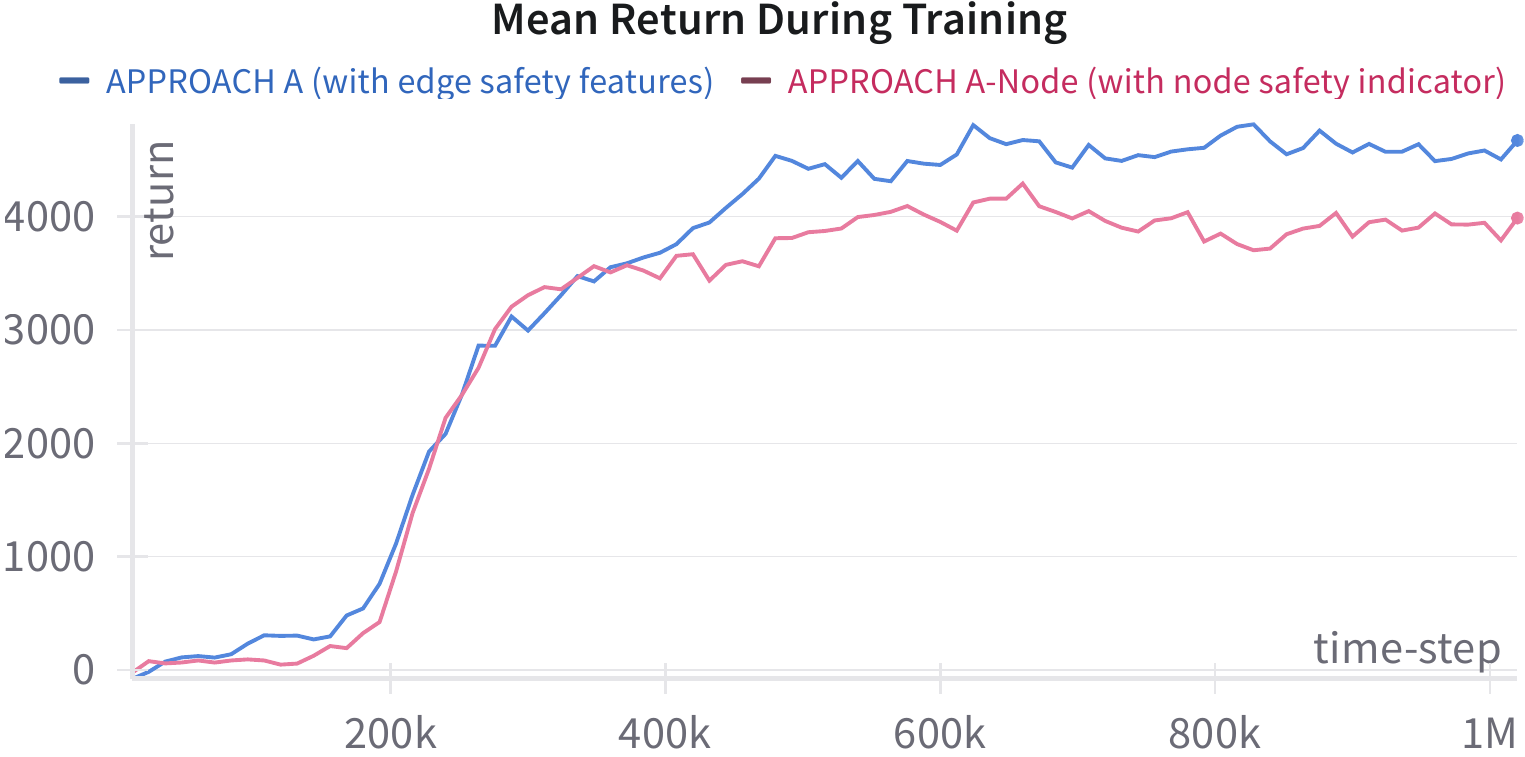}
        \caption{Mean return achieved during training. Approach A (with safety-informed edges) converges to higher rewards compared to Approach A-Node.}
        \label{subfig:training_impact}
    \end{subfigure}\\[2ex]
    \begin{subfigure}[t]{0.48\textwidth}
        \centering
        \includegraphics[scale=0.3]{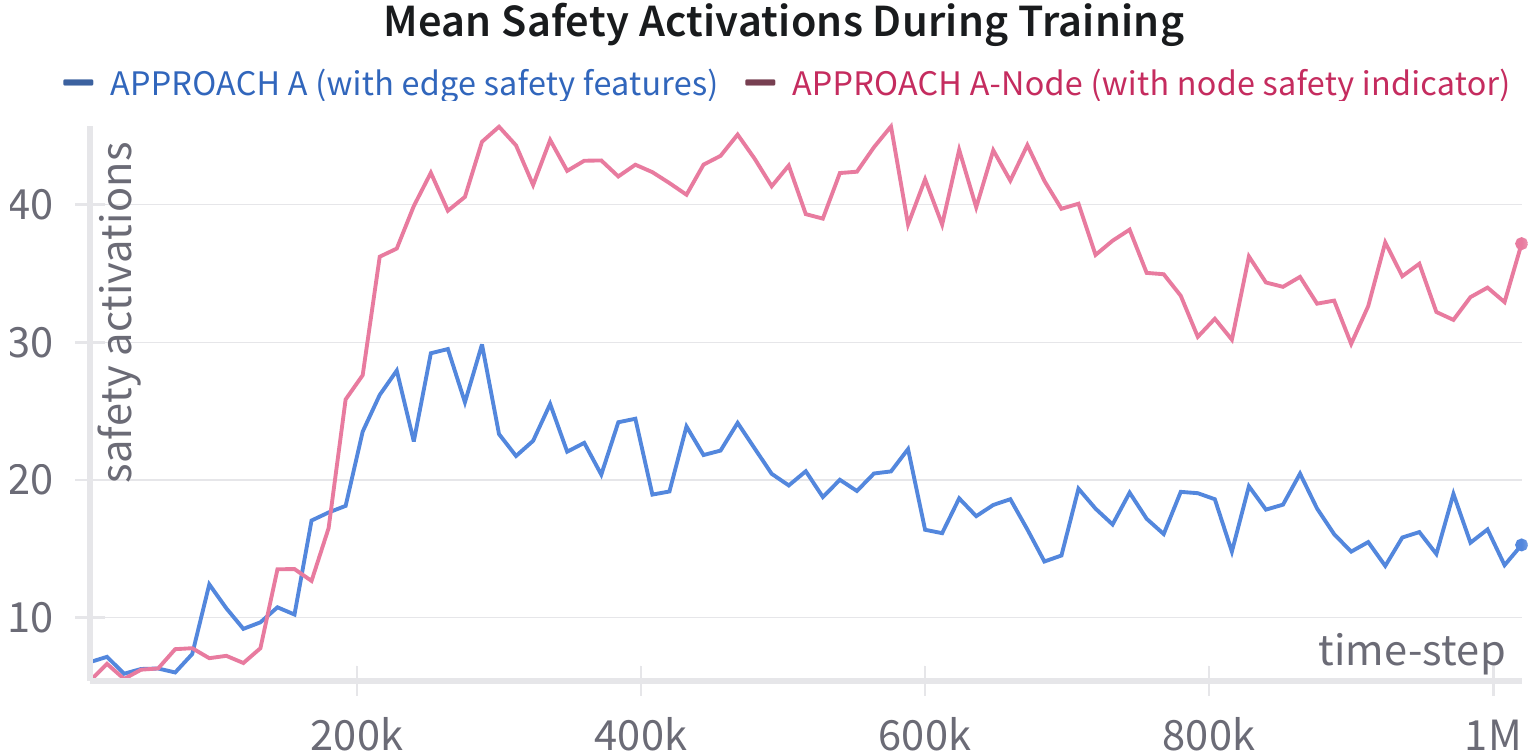}
        \caption{Average number of safety filter activations per episode. Approach A consistently triggers fewer interventions.}
        \label{subfig:safety_activations_impact}
    \end{subfigure}
    \caption{Impact of safety-informed edge features.}
    \label{fig:safety_edges_impact}
\end{figure}

Table \ref{tab:agent_performance_safety_truncation} summarizes our experimental evaluation of various learned \gls{marl} strategies under different training conditions. This comparison was conducted to evaluate the impact of our safety filter. We report the average communication coverage achieved by the agents and the average number of times the safety filter was activated during evaluation.

Furthermore, Figure~\ref{fig:safety_edges_impact} focuses on a comparison between Approach A and Approach A-Node in terms of the mean return, i.e. the sum of the rewards, achieved by the agents during training and the mean number of safety activations triggered by the agents. 

\subsection{Impact of the safety filter}
Baseline agents, trained without any safety mechanisms, achieved an average coverage of $(37.63 \pm 16.62)\%$ but required frequent interventions from the safety filter, averaging $26.99 \pm 12.31$ triggers per episode. In contrast, our Approach A, which incorporates our safety filter during training without explicit penalties, significantly improved average coverage to $(50.14 \pm 12.72)\%$ and reduced safety filter activations to $20.45 \pm 7.49$.

Introducing explicit penalties alongside the safety mechanism (Approach B) maintained competitive coverage at $(45.02 \pm 11.81)\%$ and further reduced the frequency of safety filter interventions to $15.25 \pm 7.82$. However, implementing truncation upon safety violations (Approach C) considerably decreased task performance to $(10.86 \pm 6.44)\%$, though it also minimized the safety triggers to $9.76 \pm 5.72$. 

These results demonstrate that integrating our safety filter during training and including safety-informed edges allows agents to learn better strategies for the dynamic network bridging environment. 

Results on Approach B and Approach C highlight how the introduction of explicit penalties might reduce safety interventions, which could have practical benefits in specific applications. However, this led to a decrease in task performance, which was conspicuous for Approach C, suggesting that agents might exploit the activation of the safety filter to achieve better performance.

\subsection{Impact of the safety-informed edges}
\begin{figure*}[htbp!]
    \centering
    \begin{subfigure}[b]{0.2\textwidth}
        \includegraphics[width=\textwidth]{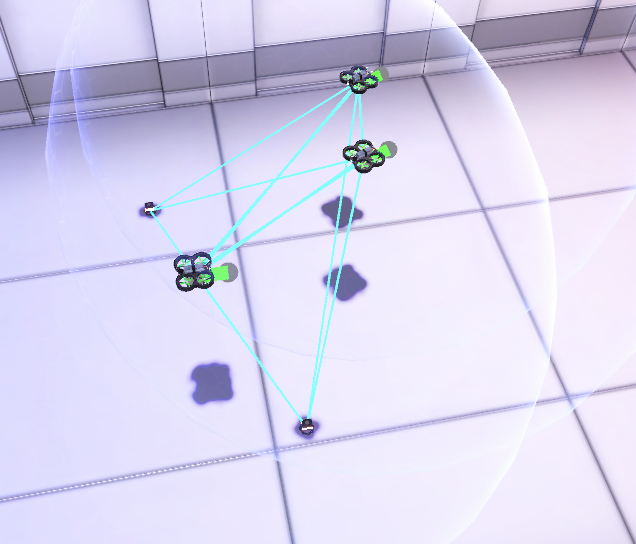}
        \label{fig:figure1}
    \end{subfigure}%
    \hspace{1em}
    \begin{subfigure}[b]{0.2\textwidth}
        \includegraphics[width=\textwidth]{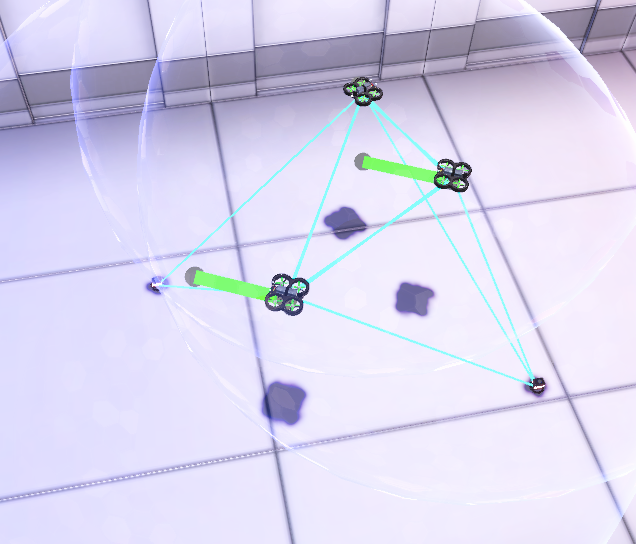}
        \label{fig:figure2}
    \end{subfigure}%
    \hspace{1em}
    \begin{subfigure}[b]{0.2\textwidth}
        \includegraphics[width=\textwidth]{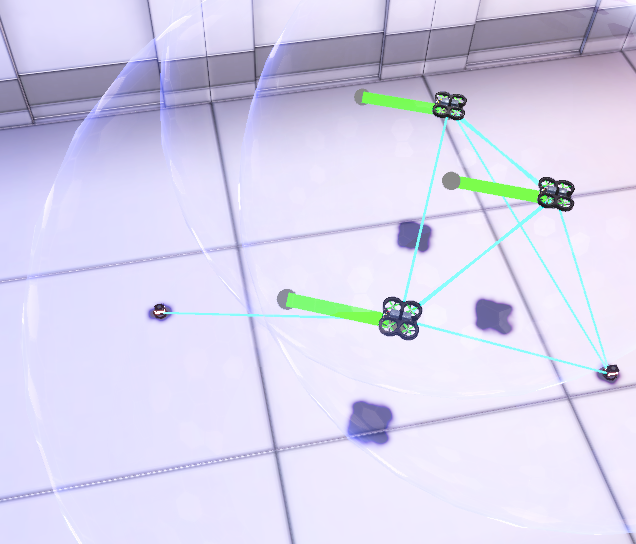}
        \label{fig:figure3}
    \end{subfigure}
    \caption{Sequence of actions executed by agents trained with Approach A when deployed in a LVC environment.}
    \label{fig:lvc}
\end{figure*}

Figure~\ref{subfig:training_impact} shows the mean return (cumulative reward) curves for Approach A and Approach A-Node, over the course of training. Approach A (with safety edge features) converges to higher returns compared to Approach A-Node, indicating that safety-informed edges facilitate more effective coordination strategies. Additionally, Figure~\ref{subfig:safety_activations_impact} presents the average frequency of safety filter activations across training episodes. The results show that Approach A triggers fewer safety interventions than Approach A-Node, particularly in the later stages of training. 

These findings suggest that edge-level safety features, rather than node-level, help achieve better performance while incurring in less safety activations. Figure~\ref{fig:lvc} shows an example of a sequence of actions operated by our decentralized agents trained with Approach A when deployed in a \gls{lvc} environment.

\section{Discussion}
Our findings demonstrate that integrating a decentralized safety filter with \gls{marl} provides a versatile solution for multi-agent networked tasks, where agents can directly communicate with their one-hop neighbors. By embedding the safety mechanism in each agent, we ensure local checks on potential collisions, thereby containing safety concerns within these immediate communication ranges. This design preserves the decentralized nature of our approach while providing global safety assurances. We applied the framework to the dynamic network bridging task, which we extended to include more realistic target mobility and non-linear \gls{uav} dynamics.

To support coordination, we included safety-informed edge features in the agents' observations. These features indicate when the safety filter is active, allowing agents to adjust their actions based on local constraints. This helped agents avoid collisions, maintain communication links, and resolve potential deadlock situations through coordination.

% Additionally, we refined the \textit{dynamic network bridging} task, incorporating more realistic mobility patterns for the targets and simulating non-linear \gls{uav} dynamics. At the same time, augmenting each agent’s observation space with safety-informed edges improves coordination. By explicitly noting safety-filter activations between neighbors, agents learn to position themselves more effectively, achieving better coverage in the bridging scenario while avoiding collisions. Agents can thus overcome potential deadlock situations—where other methods might stall—through information sharing and locally informed movement decisions.

Empirical evaluations confirm that our approach effectively improves task performance while assuring safety. Although adding explicit penalties or early-termination strategies reduced safety activations, these methods often compromised the agent's utility. In contrast, our decentralized filter and safety-informed edges guide safe exploration without requiring extensive reward engineering or outcome truncations. It is important to highlight that the design principles generalize naturally to larger numbers of agents and targets and it can be applied to different environments than the one under study here. For example, we can apply similar safety filters to other networked navigation tasks or expand the existing safety filter to include complex metrics.

Looking ahead, we plan to investigate scalability for swarms of varying sizes, refine message-passing protocols for more sophisticated risk management, and evaluate domains with adversarial behaviors or dynamic constraints. Overall, our results highlight the potential of combining local safety controllers with \gls{marl} in networked multi-agent settings, underscoring that edge-level safety signals are crucial for guiding agents toward globally safe and high-performing outcomes.

\bibliographystyle{ieeetr}
\bibliography{biblio.bib}
\end{document}